\documentclass[a4paper,USenglish]{article}

\usepackage[utf8]{inputenc}
\usepackage{amsmath}
\usepackage{amsfonts}
\usepackage{amssymb}
\usepackage{amsthm}
\usepackage{thmtools}
\usepackage{thm-restate}
\usepackage{dsfont}
\usepackage{tabularx}
\usepackage{booktabs}
\usepackage{tikz}
\usetikzlibrary{positioning,backgrounds,patterns,calc,matrix,shapes,decorations.pathreplacing,decorations.pathmorphing}
\usepackage{todonotes}
\usepackage[numbers]{natbib} 
\usepackage[pagebackref]{hyperref}
\usepackage[capitalize,noabbrev]{cleveref}
\usepackage{a4wide}

\title{The Complexity of Gerrymandering Over Graphs:\\ Paths and Trees\footnote{This work was initiated at the research retreat of the Algorithmics and Computational Complexity group held in September~2020 in Zinnowitz, Germany.}}

\author{Matthias Bentert \and Tomohiro Koana\footnote{Supported by the Deutsche Forschungsgemeinschaft (DFG), project FPTinP, NI 369/19.} \and Rolf Niedermeier}

\date{TU Berlin, Algorithmics and Computational Complexity, Germany\\\texttt{\small \{matthias.bentert,tomohiro.koana,rolf.niedermeier\}@tu-berlin.de}}

\newcommand{\probDef}[4]{
\begin{center}   
	\fbox{~\begin{minipage}{.95\textwidth}
		\vspace{2pt} 

		\noindent
		\normalsize\textsc{#1}
		
		\vspace{4pt}
		\setlength{\tabcolsep}{3pt}
		\renewcommand{\arraystretch}{1.0}
		\begin{tabularx}{\textwidth}{@{}lX@{}}
			\normalsize\textbf{In:} 	& \normalsize#2 \\
			\normalsize\textbf{#4:} 	& \normalsize#3
		\end{tabularx}
	\end{minipage}}
\end{center}
}

\newcommand{\decProb}[3]{\probDef{#1}{#2}{#3}{?}}

\newtheorem{theorem}{Theorem}

\DeclareMathOperator{\col}{col}
\newcommand{\N}{\mathds{N}}
\newcommand{\colo}[1]{\ensuremath{c_{#1}}}
\newcommand{\defeq}{:=}

\newcommand{\GerryJ}{{\normalfont\textsc{Gerrymandering over Graphs}}}

\begin{document}
\maketitle

\begin{abstract}
	Roughly speaking, 
	gerrymandering is the systematic manipulation of the boundaries of electoral 
	districts to make a specific (political) party win as many districts as possible.
    While typically studied from a geographical point of view, addressing 
	social network structures, the investigation 
	of gerrymandering over graphs was recently initiated by Cohen-Zemach et al.~[AAMAS~2018].
	Settling three open questions of Ito et al.~[AAMAS~2019], we classify the 
    computational complexity of the NP-hard problem \GerryJ{} when restricted to 
	paths and trees. Our results, which are mostly of negative nature (that is, 
	worst-case hardness), in particular yield two complexity dichotomies 
	for trees. For instance, the problem is polynomial-time solvable for two 
	parties but becomes weakly NP-hard for three.
	Moreover, we show that the problem remains NP-hard even when the input graph is a path.
\end{abstract}

\section{Introduction}
How to influence an election? One answer to this is gerrymandering~\cite{CD19,GKL19,Tas11}.
Gerrymandering is the systematic manipulation of the boundaries of 
electoral districts in favor of a particular party. 
It has been studied in the political sciences for decades~\cite{Sch87}.
In recent years, various models of gerrymandering were investigated from 
an algorithmic and computational perspective.
For instance, Lewenberg et al.\;\cite{LLR17} and Eiben et al.\;\cite{EFPS20}
studied the (parameterized) computational complexity of gerrymandering assuming that 
the voters are points in a two-dimensional space and the task is to place 
$k$~polling stations where each voter is assigned to the polling station closest to her.
Cohen-Zemach et al.\;\cite{ZLR18} introduced a version of gerrymandering over graphs (which may be seen 
as models of social networks) where the question is 
whether a given candidate can win at least~$\ell$ districts.
This leads to the question whether there is a partition of the graph
into $k$~connected subgraphs 
such that at least~$\ell$ of these are won by a designated candidate; herein, 
$k$~and~$\ell$ are part of the input of the computational problem. 
Cohen-Zemach et al.\ showed that this version is NP-complete even when restricted to planar graphs.
Following up on the pioneering work of Cohen-Zemach et al.\;\cite{ZLR18}, Ito et al.\;\cite{IKK19} performed a refined 
complexity analysis, particularly taking into account the special graph structures
of cliques, paths, and trees. Indeed, their formal model is slightly different 
from the one of Cohen-Zemach et al.\;\cite{ZLR18} and their work will be our main point of reference.
Notably, both studies focus on the perhaps simplest voting rule, Plurality. 

We mention in passing that earlier work also studied the special case of 
gerrymandering on grid graphs. More specifically, Apollonio et al.\;\cite{ABL09} analyzed gerrymandering in grid graphs where each district in the solution has to be of (roughly) the same size and they analyzed, focusing on two candidates (equivalently, two parties), the maximum possible win margin if the two candidates had the same amount of support. 
Later, Borodin et al.\;\cite{BLS18} also considered gerrymandering on grid graphs with two parties (expressed by colors red and blue), but here each vertex represents a polling station and thus is partially ``red'' and partially ``blue'' colored. They provided a worst-case analysis for a two-party situation in terms of the total fraction of votes the party responsible for the gerrymandering process gets. They also confirmed their findings with experiments.

To formally define our central computational problem, we continue with a few
definitions.
For a vertex-colored graph and for each color~$r$, let~$V_r$ be the set of~\mbox{$r$-colored} vertices. 
A vertex-weighted graph is~\emph{$q$-colored} if for each color~$r$ it holds that~${\sum_{v\in V_q} w(v) \geq \sum_{v\in V_r} w(v)}$.
A vertex-weighted graph is \mbox{\emph{uniquely~$q$-colored}} if~${\sum_{v\in V_q} w(v) > \sum_{v\in V_r} w(v)}$ for each color~$r \neq q$.
Thus, we arrive at the central problem of this work, going back to Ito et al.\;\cite{IKK19}.

\decProb{\GerryJ}
{An undirected, connected graph~$G = (V, E)$, a weight function~${w \colon V \rightarrow \N}$, a set~$C$ of colors, a target color~${p \in C}$, a coloring function~${\col \colon V \rightarrow C}$, and an integer~$k$.}
{Is there a partition~$\mathcal V$ of~$V$ into exactly~$k$ subsets~${V_1, \dots, V_k}$ such that every ${V_i \in \mathcal{V}}$,  ${i \in [k]}$, induces a connected subgraph in~$G$ and the number of uniquely~$p$-colored induced subgraphs exceeds the number of~$r$-colored induced subgraphs for each~${r \in C \setminus \{ p\}}$?}

\cref{fig:gerrymanderingExample} presents a simple example of \GerryJ{}. 
We remark that all our results except for \Cref{thm:pathnphard} (that is, the NP-hardness on paths) also transfer to 
the slightly different model of Cohen-Zemach et al.\;\cite{ZLR18}.\footnote{In fact, we conjecture that the gerrymandering problem of Cohen-Zemach et al.\;\cite{ZLR18} is polynomial-time solvable on paths.}

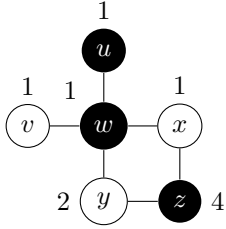
\begin{figure}[t]
	\centering
	\begin{tikzpicture}
		\node[circle, fill=black,label=$1$] at (0,1) (u) {\color{white} $u$};
		\node[circle,draw,label=$1$] at (-1,0) (v) {$v$};
		\node[circle,fill=black,label=above left:$1$] at (0,0) (w) {\color{white} $w$} edge (u) edge (v);
		\node[circle,draw,label=$1$] at (1,0) (x) {$x$} edge (w);
		\node[circle,draw,label=left:$2$] at (0,-1) (y) {$y$} edge (w);
		\node[circle,fill=black,label=right:$4$] at (1,-1) (z) {\color{white} $z$} edge (x) edge (y);
	\end{tikzpicture}
	\caption{An example input instance for \GerryJ{} with two colors~(black and white) where black is the target color and where the numbers next to the vertices illustrate the vertex weights. For~$k=2$, a solution for this instance is~$\mathcal{V} = \{\{u,v,w\},\{x,y,z\}\}$ as each of these two parts induces a uniquely black-colored connected subgraph.}
	\label{fig:gerrymanderingExample}
\end{figure}
We also use an equivalent interpretation of 
solution partitions~$\mathcal V$ for \GerryJ{}.
Since each part~$V_i \in \mathcal{V}$ has to induce a connected subgraph, in the spirit of edge deletion problems from algorithmic graph theory, we also represent solutions by a set of edges such that removing these yields the disjoint union of subgraphs induced by each part~$V_i \in \mathcal{V}$.
In \cref{fig:gerrymanderingExample}, removing the edges~$\{w,x\}$ and~$\{w,y\}$ yields a solution.

Finally, regarding notation, for a color~$q$ we use~${w_q(v) \defeq w(v)}$ if~$v$ is of color~$q$ and~${w_q(v) = 0}$ if~$v$ has another color.
Further, we use~${w(V') \defeq \sum_{v \in V'} w(v)}$ and~${w_q(V') \defeq \sum_{v \in V'} w_q(v)}$.

\begin{table}[t]
	\centering
	\caption{Results overview. The diameter of a graph is denoted by~$\text{diam}$.}
	\label{tab-res}
	\begin{tabular}{clll}
		\toprule
		\ \ \ \ \ \ \ \ \ \ \ \ \ & Restriction\ \ \ \ & Complexity\ \ \ \ \ \ \ \ \ \ \ \ \ \ \ \ \ \ \ \ & Reference \ \ \ \ \ \ \ \ \ \ \\
		\midrule
		Paths & no & NP-hard & \Cref{thm:pathnphard} \\
		& constant $|C|$ & polynomial time & \cite[Theorem 4.4]{IKK19} \\
		\midrule
		Trees & constant $|C|$ & pseudo-polynomial time & \cite[Theorem 4.5]{IKK19} \\
		& $|C| = 2$ & polynomial time & \Cref{thm:treeP} \\
		& $|C| \ge 3$ & weakly NP-hard & \Cref{TreeNPhard} \\ [0.2ex]
		& $\text{diam} = 2$ & polynomial time & \cite[Theorem 4.2]{IKK19} \\
		& $\text{diam} = 3$ & polynomial time & \Cref{thm:diam3} \\
		& $\text{diam} \geq 4$ & NP-hard & \cite[Theorem 3.3]{IKK19} \\
		\bottomrule
	\end{tabular}
\end{table}

\paragraph{Known and new results.}
As mentioned before, we essentially build our studies on the work of 
Ito et al.\;\cite{IKK19}, in particular studying exactly the same computational
problem. We only focus on the case of path and tree graphs as input, 
whereas they additionally studied cliques. For cliques, they showed 
NP-hardness already  for $k=2$ and two colors. On the positive side, for cliques 
they provided a pseudo-polynomial-time algorithm for~$k=2$ and and a polynomial-time 
algorithm for each fixed~$k\geq 3$.
Moving to paths and trees, besides some positive algorithmic and hardness results 
Ito et al.\;\cite{IKK19} particularly left three open problems: 
\begin{enumerate}
\item Existence of a polynomial-time algorithm for paths when~$|C|$ is part of the input.
\item Existence of a polynomial-time algorithm for trees when~$|C|$ is a constant. 
\item Existence of a polynomial-time algorithm for trees of diameter exactly three.
\end{enumerate}
Indeed, they called the first two questions the ``main open problems'' of their paper.
We settle, all three questions, the first two in the negative by showing NP-hardness.
See \cref{tab-res} for an overview on some old and our new results.
Notably, our new results (partially together with the 
previous results of Ito et al.\;\cite{IKK19}) reveal two sharp complexity dichotomies for trees.  
For up to two colors, the problem is polynomial-time solvable, whereas it gets 
NP-hard with three or more colors; moreover, it is polynomial-time solvable for 
trees with diameter at most three but NP-hard for trees with diameter at least four.
In the remainder of this work, we first present our results for paths, 
and then for trees.

\section{NP-hardness on paths}

Ito et al.\;\cite{IKK19} showed that \GerryJ{} on paths can be solved in polynomial time for fixed $|C|$, and left open the question of polynomial-time solvability on paths when $|C|$ is unbounded.
Negatively answering their question, we show that \GerryJ{} remains NP-hard on paths even if every vertex has unit weight.

\begin{theorem}
	\label{thm:pathnphard}
	\GerryJ{} restricted to paths is NP-hard even if all vertices have unit weight.
\end{theorem}

\begin{proof}
	We reduce from \textsc{Clique} on regular graph, which is NP-hard~\cite{MS12}.
	Let~$(G, \ell)$ be an instance of \textsc{Clique}, where~$G$ is~$d$-regular for some integer $d$, and~$\ell$ is the sought solution size.
	The main idea is to first construct an equivalent instance of \GerryJ{} where the graph consists of disjoint paths.
	Afterwards, we slightly modify the reduction to obtain one connected path.

	All vertices in the following constructions have weight one.
	Let~$n$ and $m$ be the number of vertices and edges in $G$, respectively, and let $N \defeq 4n^2$. 
	We introduce a path $P_v$ on $4N - 1$ vertices for each vertex~$v \in V$ and a path~$P_e'$ on four vertices for each edge~$e \in E$.
	Moreover, we introduce an independent set $S$ of $2N - (n - \ell) + 1$ vertices.
	We denote by~$G'=(V',E')$ the disjoint union of all~$P_v$ for~$v \in V$, all~$P_e'$ for~$e \in E$, and~$S$.
	Note that~$G'$ has~$z \defeq 2N + \ell + m + 1$ connected components.

	We introduce colors $p, q, r$, and a unique color $\colo{v}$ for each~$v \in V$, where~$p$ is the target color.
	We color~$N + 1$ vertices of $S$ with color~$p$ and $N - (n - \ell)$ vertices of $S$ with color~$q$.
	For each vertex $v \in V$, we color the vertices in $P_v$ as follows.
	\begin{itemize}
		\item The first $N - 1$ vertices receive color~$q$,
		\item for each~$i \in [N]$, the~$(N - 1 + 3i)$-th vertex receives color~$\colo{v}$, and
		\item each remaining vertex receives a new color (which is distinct for each vertex).
	\end{itemize}
	An illustration of the path $P_v$ is shown in \cref{fig:vertexselectiongadget}.
	For each edge~$e = \{u,v\} \in E$, we color the two inner vertices of $P'_e$ with color~$r$ and the endpoints with colors~$\colo{u}$ and $\colo{v}$, respectively.
	Finally, we set~$k \defeq (n - \ell) \cdot 3N + d \ell + \binom{\ell}{2} + z$.

	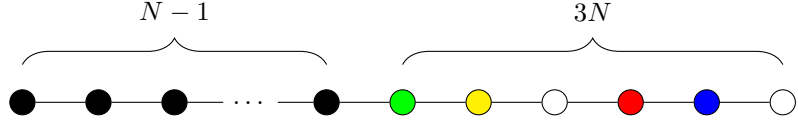
\begin{figure}
		\centering
		\begin{tikzpicture}
			\node[circle, fill=black, draw] at (0,0) (q1) {};
			\node[circle, fill=black, draw] at (1,0) (q2) {} edge (q1);
			\node[circle, fill=black, draw] at (2,0) (q3) {} edge (q2);
			\node at(3,0) (qdots) {$\dots$} edge (q3);
			\node[circle, fill=black, draw] at (4,0) (q4) {} edge (qdots);
			\node[circle, fill=green, draw] at (5,0) (v11) {} edge (q4);
			\node[circle, fill=yellow, draw] at (6,0) (v12) {} edge (v11);
			\node[circle, fill=white, draw] at (7,0) (v13) {} edge (v12);
			\node[circle, fill=red, draw] at (8,0) (v21) {} edge (v13);
			\node[circle, fill=blue, draw] at (9,0) (v22) {} edge (v21);
			\node[circle, fill=white, draw] at (10,0) (v23) {} edge (v22);
						
			\draw [decorate,decoration={brace,amplitude=10pt},yshift=0pt] (0,.5) -- (4,.5);
			\node at(2,1.2) {$N-1$};
			\draw [decorate,decoration={brace,amplitude=10pt},yshift=0pt] (5,.5) -- (10,.5);
			\node at(7.5,1.2) {$3N$};
		\end{tikzpicture}
		\caption{An example of a gadget for a vertex~$v$. Black and white vertices represent~$q$-colored and~$\colo{v}$-colored vertices, respectively, and each other vertex has a distinct color.}
		\label{fig:vertexselectiongadget}
	\end{figure}
	
	First, we show that if $G$ contains a clique $K$ of size~$\ell$, then the constructed instance of \GerryJ{} is a yes-instance.
	We will specify the set~$E''$ of exactly~$k - z$ edges such that the connected components of~${G'' = (V', E' \setminus E'')}$ correspond to a solution.
	Note that each removal of an edge increases the number of connected components by exactly one.
	\begin{itemize}
		\item
			For each vertex~$v \in V \setminus K$, the edge set~$E''$ contains all $3N$ edges in $P_v$ that are not between two~$q$-colored vertices.
			There are $(n - \ell) \cdot 3N$ such edges.
		\item
			For each vertex~$v \in K$ and each edge~$e = \{u,v\}$, the edge set~$E''$ contains the edge incident to the~$\colo{v}$-colored vertex in $P'_e$.
			There are~$d \ell$ such edges as each vertex in the input graph has~$d$ neighbors.
		\item
			For each edge~$e$ where both endpoints are contained in~$K$, the edge set~$E''$ contains the edge between the two inner ($r$-colored) vertices in $P'_e$.
			There are~$\binom{\ell}{2}$ such edges.
	\end{itemize}
	Thus, $E''$ contains $(n-\ell)\cdot 3N + d \ell + \binom{\ell}{2} = k - z$ edges in total, leaving~$k$ connected components in the graph $G''$.

	Now we examine the color of each connected component of $G''$.
	First, note that there are~$N+1$ connected components that are uniquely~$p$-colored.
	We now show that for each color~$c$ other than~$p$ there are at most~$N$ connected components which are~$c$-colored.
	\begin{itemize}
		\item
			For color $q$, observe that there are~$N-(n-\ell)$ isolated vertices of color~$q$ in~$S$ and for each vertex~${v \in V \setminus K}$ there is exactly one~$q$-colored connected component contained in $P_v$ and for every vertex~$v \in K$ there is no~$q$-colored connected component in~$P_v$.
			Hence, there are~$N - (n-\ell) + (n-\ell) = N$ connected components that are~$q$-colored.
		\item
			For color $r$, note that there are $2m < N$ vertices which are~$r$-colored. Thus, there are less than~$N$ connected components that are~$r$-colored.
		\item 
			For each color~$\colo{v}$ with $v \in V \setminus K$, there are~$N$ connected components in~$P_v$ that are~$\colo{v}$-colored.
			All other~\mbox{$\colo{v}$-colored} vertices are contained in $P'_e$ for some~${e \in E}$ and those belong to $r$-colored component by construction.
			Hence, there are~$N$ connected components that are~$\colo{v}$-colored.
		\item
			For each color $\colo{v}$ with $v \in K$, the whole path~$P_v$ remains one connected component which is~$\colo{v}$-colored.
			All other~$\colo{v}$-colored vertices are contained in~$P_e'$ for some $e \in E$ and since~${N > m}$, there are at most~$N$ connected components that are~$\colo{v}$-colored.
	\end{itemize}
	Thus, if~$G$ contains a clique of size~$\ell$, then the constructed instance is a yes-instance.

	Conversely, we show that if the constructed instance of \GerryJ{} has a solution $\mathcal{V}$, then there is a clique of size~$\ell$ in $G$.
	Let~$E''$ be a set of exactly~$k - z$ edges in~$G'$ such that the connected components of~${G'' = (V', E' \setminus E'')}$ correspond to~$\mathcal{V}$.
	Let $J$ be the set of vertices~$v \in V$ such that $P_v$ contains an edge of~$E''$ and let~$K \defeq V \setminus J$.
	For each vertex~$v \in J$, let $n_v^q$ and $n_v^c$ be the number of connected components of $P_v - E''$ which are~\mbox{$q$-colored} and~\mbox{$v_c$-colored}, respectively.
	Our goal is to show that $K$ forms a clique of size~$\ell$ in~$G$.
	To this end, we derive an upper bound on the size of $E''$ in terms of~$n_v^q$,~$n_v^c$, and~$|J|$:
	\begin{enumerate}
		\item
			For each vertex~$v \in J$, there are at most~${n_v^q - 1}$~edges in~$P_v$ whose endpoints are~$q$-colored.
			Since there are~$N + 1$ isolated~$p$-colored vertices and~${N - (n - \ell)}$ isolated~$q$-colored vertices in~$S$, it follows that~${\sum_{v \in J} n_v^q \le n - \ell}$.
			Thus, $E''$ contains at most $\sum_{v \in J} n_v^q - 1= n - \ell - |J|$ edges in $P_v$ both of whose endpoints are $q$-colored.
		\item
			For each vertex~$v \in J$, the edge set~$E''$ contains at most~$3 n_v^c$ edges in~$P_v$ where at least one endpoint is not~\mbox{$q$-colored}.
		\item
			For each vertex~$v \in J$, the edge set~$E''$ contains at most~$N - n_v^c$ edges incident to a~$c_v$-colored endpoint in a~$P'_e$ for some edge~$e \in E$.
		\item
			For each vertex~$v \in K$, there are exactly~$d$ edges incident to a~$c_v$-colored endpoint that are contained in a~$P'_e$ for some edge~$e \in E$.
			Thus, $E''$ contains at most~${d \cdot |K| = d \cdot (n - |J|)}$ such edges.
		\item
			Finally, we consider edges between inner vertices of $P'_e$ for $e \in E$.
			Observe that if such an edge $e = \{ u, v \}$ is in $E''$, then $G''$ has one $\colo{u}$-colored component and one $\colo{v}$-colored component.
			Thus, $|E''|$ contains at most $${\binom{|K|}{2} + \left( \sum_{v \in J} N - n_v^c \right) = \binom{n - |J|}{2} + \sum_{v \in J} N - n_v^c}$$ such edges.
	\end{enumerate}
	Summing over these edges yields that $E''$ contains at most
	\begin{align*}
		& (n - \ell - |J|) + \left( \sum_{v \in J} 3 n_v^c \right) + \left( \sum_{v \in J} N - n_v^c \right) \nonumber\\
		&+\ d\cdot (n - |J|) + \binom{n - |J|}{2} + \left( \sum_{v \in J} N - n_v^c \right) \nonumber \\
		&\le (n - \ell - |J|) + 3N \cdot |J| \nonumber
		+\ d \cdot (n - |J|) + \binom{n - |J|}{2} \nonumber
	\end{align*}
	edges.
	Here, the inequality is due to the fact that~$n_v^c \le N$.
	Thus,~${|E''| \le f(x)}$, where $$f(x) \defeq (n - \ell - x) + 3N \cdot x + d \cdot (n - x) + \binom{n - x}{2}.$$

	Next we show that~$|J| \le n - \ell$.
	Recall that $G$ has $N + 1$ isolated $p$-colored vertices and $N - (n - \ell)$ isolated $q$-colored vertices.
	Since the path~$P_v$ contains at least one~$q$-colored part for every vertex~$v \in J$, we obtain $|J| \le n - \ell$.

	Notice that $f(x)$ is monotonically increasing for $x \ge 0$ and that from this follows that~$k - z = |E''| \le f(|J|) \le f(n - \ell)$.
	Note that~${f(n - \ell) = k - z}$ by the definition of $f$.
	Consequently, we have~${f(|J|) = f(n - \ell)}$ and hence $|J| = n - \ell$.
	Finally, note that for any solution where~$|J| = n - \ell$, we cannot remove any edges between two~$q$-colored vertices (as this would result in at least~$N+1$ connected components that are~$q$-colored).
	Hence,~$n_v^q = 1$ for each~$v \in V$ and thus summing up all edges in~$E''$ without the edges between two~$q$-colored vertices yields
	$$|E''| \leq \left( \sum_{v \in J} 2N + n_v^c \right) + d \ell + \binom{\ell}{2}.$$
	For~$E''$ to contain~$k-z = (n  - \ell) \cdot 3N + d \ell + \binom{\ell}{2}$ edges, it has to also hold that~$n_v^c = N$ for each vertex~$v \in J$.
	Hence, there are exactly~$\binom{\ell}{2}$ edges in~$E''$ between two~$r$-colored vertices in~$P'_e$ for edges~$e \in E$.
	Note that for each such edge~$e \in E$ it has to hold that both endpoints of~$e$ are in~$K$ as otherwise there are~$N+1$ connected components in~$G''$ of color~$\colo{v}$ (where~$v \in J$ is an endpoint of~$e$).
	Thus, there are~$\ell$ vertices in~$K$ that share~$\binom{\ell}{2}$ edges between them, that is,~$K$ induces a clique of size~$\ell$.

	We next show how to connect the different paths of the construction to obtain a single connected path.
	For~${M \defeq 4Nn + 3m}$, we simply add a path of~$M$ vertices between each connected component in the previous reduction (that results in multiple disconnected paths) where each vertex has a unique color.
	Note that there are exactly~$z-1$ such paths and thus in total~$(z-1) \cdot (M+1)$ new edges.
	Finally, we set~$k \defeq (n-\ell)\cdot 3N + \ell \cdot d + \binom{\ell}{2} + (z-1) \cdot (M+1)$.
	The correctness of this adaption is straight-forward:
	If there is a solution for the instance consisting of multiple paths, then removing the newly introduced edges clearly gives a solution for the new instance consisting of a single path.
	If there is no solution for the instance consisting of multiple paths, then note that since~$M$ is larger than the number of edges in the original construction and~$k > (z-1) \cdot (M+1)$, at least one edge from each newly introduced path is removed.
	Hence, vertices that are in different connected components in the original construction are also in different connected components in any solution.
	Moreover, since all newly introduced vertices have unique colors and all vertices have the same weight, any color of a connected component in a solution for the instance consisting of multiple paths also has the same color in the newly constructed instance.
\end{proof}

In the above reduction,
we use an unbounded number of colors. 
This appears to be inevitable since \GerryJ{} is polynomial-time solvable for any constant $|C|$.
We wonder whether there are other graph classes for which \GerryJ{} can be solved in polynomial time when $|C|$ is constant.
Caterpillars form a possible candidate.

\section{Complexity on trees}
In this section, we first address the special case of three colors (NP-hard), 
then two colors (polynomial-time solvable), 
and finally we discuss the polynomial-time solvability for diameter-three trees.

Ito et al.\;\cite{IKK19} developed a pseudo-polynomial time algorithm for \GerryJ{} on trees for constant~$|C|$, which led them to ask whether it is also polynomial-time solvable for fixed~$|C|$.
We show that \GerryJ{} on trees is weakly NP-hard even if~$|C| = 3$, answering their question in the negative.
In the following subsection, we will then show the polynomial-time solvability for $|C| = 2$. So we have a tight classification.

\begin{theorem}\label{TreeNPhard}
	\GerryJ{} restricted to trees is weakly NP-hard even if $|C| = 3$.
\end{theorem}
\begin{proof}
	We reduce from \textsc{Partition}, which is known to be NP-hard \cite{GJ79}.
	Given a multi-set~$A$ of~$n$ non-negative integers, the task is to find a subset~$B \subseteq A$ of exactly $n / 2$ integers whose sum is $s/2$, where~$s \defeq \sum_{a\in A} a$.
	We can assume that~$s$ is a multiple of $n$ (otherwise we multiply each element of $A$ by $n$).
	Let $N \defeq s + 1$ and let $M$ be some natural number greater than $N \cdot 2^{n} (n + 1) + s / 2 + 1$.
	For the construction, we use a set~$C=\{p,q,r\}$ of three colors, where~$p$ is the target color.
	We start with a star with a center vertex~$z$ and a set~$L$ of $n/2$~leaves.
	We color every vertex in the star with color~$p$.
	We assign the weights~$w(z) \defeq Mn + s / 2 + 1$ to the center $z$ and~$w(\ell) \defeq1$ for each leaf $\ell \in L$.
	For each~$i \in [n]$, we do the following.
	\begin{itemize}
		\item
			We introduce two vertices~$x_{i}^{q}$ and~$y_{i}^{q}$ of color~$q$ and two vertices~$x_{i}^{r}$ and~$y_{i}^{r}$ of color~$r$.
			Let~$X_i \defeq \{ x_i^q, x_i^r \}$,~${Y_i \defeq \{ y_i^q, y_i^r \}}$, and~$Z_i\defeq X_i \cup Y_i$.
		\item
			We add four edges~$\{z,x_i^q\}$, $\{z,y_i^q\}$, $\{x_i^q,x_i^r\}$, and~$\{y_i^q,y_i^r\}$.
		\item
			We define the weights for each vertex in $Z_i$ as
			\begin{align*}
				w(x_i^q) &\defeq M + N \cdot 2^i + a_i,\\
				w(x_i^r) &\defeq M - N \cdot 2^i, \\
				w(y_i^r) &\defeq M + N \cdot 2^i - a_i + \frac{2s}{n}, \text{ and} \\
				w(y_i^q) &\defeq M - N \cdot 2^i.
			\end{align*}
			Observe that the weights are integral since~$s$ is divisible by~$n$.
			In addition, observe that~$X_i$ is~$q$-colored and that~$Y_i$~is \mbox{$r$-colored}.
	\end{itemize}
	Illustrating the constructed graph is depicted in \cref{fig:NPTree}.
	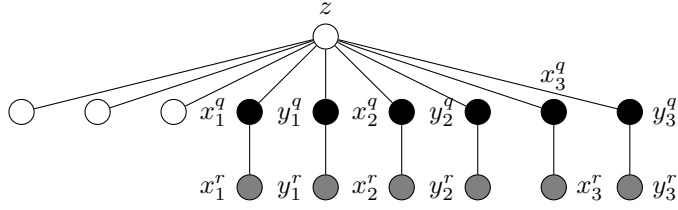
\begin{figure}[t]
		\centering
		\begin{tikzpicture}
			\node[circle, draw,label=$z$] at (0,1) (z) {};
			\node[circle, draw] at (-4,0) (l1) {} edge (z);
			\node[circle, draw] at (-3,0) (l2) {} edge (z);
			\node[circle, draw] at (-2,0) (l3) {} edge (z);
			\node[circle, draw, fill=black,label=left:$x_1^q$] at (-1,0) (x1) {} edge (z);
			\node[circle, draw, fill=black,label=left:$y_1^q$] at (0,0) (y1) {} edge (z);
			\node[circle, draw, fill=black,label=left:$x_2^q$] at (1,0) (x2) {} edge (z);
			\node[circle, draw, fill=black,label=left:$y_2^q$] at (2,0) (y2) {} edge (z);
			\node[circle, draw, fill=black,label=above:$x_3^q$] at (3,0) (x3) {} edge (z);
			\node[circle, draw, fill=black,label=right:$y_3^q$] at (4,0) (y3) {} edge (z);
			\node[circle, draw, fill=gray,label=left:$x_1^r$] at (-1,-1) () {} edge (x1);
			\node[circle, draw, fill=gray,label=left:$y_1^r$] at (0,-1) () {} edge (y1);
			\node[circle, draw, fill=gray,label=left:$x_2^r$] at (1,-1) () {} edge (x2);
			\node[circle, draw, fill=gray,label=left:$y_2^r$] at (2,-1) () {} edge (y2);
			\node[circle, draw, fill=gray,label=right:$x_3^r$] at (3,-1) () {} edge (x3);
			\node[circle, draw, fill=gray,label=right:$y_3^r$] at (4,-1) () {} edge (y3);
		\end{tikzpicture}
		\caption{An illustration of the construction in the proof of \cref{TreeNPhard}. White vertices represent~$p$-vertices, black vertices represent~$q$-colored vertices, and gray vertices represent~$r$-colored vertices.}
		\label{fig:NPTree}
	\end{figure}
	Clearly, the constructed graph $G = (V, E)$ is a tree.
	To conclude the construction of the \GerryJ{} instance, we set $k \defeq 3n / 2 + 1$.
	
	We next show that the construction is correct.
	Suppose that there is a subset~$I \subseteq [n]$ of size exactly $n / 2$ such that~$\sum_{i \in I} a_i = s / 2$.
	Then, the partition
	\begin{align*}
		\mathcal{V} = \{ V' \} \cup \{ \{ \ell \} \mid \ell \in L \}\} \cup \{ X_i \mid i \in [n] \setminus I\} \cup \{ Y_i \mid i \in [n] \setminus I \},
	\end{align*}
	where $V' \defeq \{\{ z \} \cup \{ v \mid \exists i \in I.\;v \in X_i \cup Y_i\}\}$ is a solution for the constructed instance of \GerryJ:
	First, observe that $V'$ is $p$-colored as~${w_p(V') = Mn + s / 2 + 1}$ and~${w_q(V') = w_r(V') = Mn + s / 2}$.
	We also observe that the singleton $\{ \ell \}$ is $p$-colored for each leaf $\ell \in L$, and hence~$\mathcal{V}$ has $n / 2 + 1$ subsets which are $p$-colored.
	Since $X_i$ is $q$-colored and $Y_i$ is $r$-colored for each~$i \in [n]$, exactly~$n - |I| = n / 2$ subsets of~$\mathcal{V}$ are~$q$-colored and exactly~$n - |I| = n / 2$ subsets of~$\mathcal{V}$ are~$r$-colored.
	Thus, $\mathcal{V}$ is indeed a solution.

	Conversely, suppose that there is a solution~$\mathcal{V}$.
	We show that the \textsc{Partition} instance is a yes-instance.
	Note that there are at least $k / |C| = (3n / 2 + 1) / 3 > n / 2$ parts in~$\mathcal{V}$ which are uniquely $p$-colored.
	Since there are exactly $n / 2 + 1$ vertices of color $p$, each vertex of color~$p$ is contained in a distinct part in~$\mathcal{V}$.
	In particular, this means that $\{ \ell \} \in \mathcal{V}$ for each leaf~$\ell \in L$.

	Let $V_z \in \mathcal{V}$ denote the subset containing the center $z$, and let~$n_q$ and $n_r$ denote the number of vertices of color~$q$ and~$r$ in~$V_z$, respectively.
	As each vertex of color~$q$ or~$r$ has weight at least~$M - N \cdot 2^n$, we have~${w_q(V_z) \ge (M - N \cdot 2^n) \cdot n_q}$ and~${w_r(V_z) \ge (M - N \cdot 2^n) \cdot n_r}$.
	Since~$V_z$ is uniquely~$p$-colored, we have
	\begin{align*}
		\max \{ w_q(V_z), w_r(V_z) \} < w_p(V_z) &= w(v) = Mn + s / 2 + 1 \text{ and}\\
		\max \{ n_q, n_r \} < \frac{Mn + s / 2 + 1}{M - N \cdot 2^n} &= n + \frac{N \cdot 2^n n + s / 2 + 1}{M - N \cdot 2^n} < n + 1.
	\end{align*}
	Here, the last inequality follows since ${M > N \cdot 2^n (n + 1) + s / 2 + 1}$.
	Thus,~$|V_z|$ contains at most $n_q + n_r + 1 \le 2n + 1$ vertices.

	Let $\mathcal{V}' \defeq \mathcal{V} \setminus (\{V_z\} \cup \{\{ \ell \} \mid \ell \in L \})$ be the collection of subsets of~$\mathcal{V}$ not containing any $p$-colored vertices.
	Notice that~$|\mathcal{V}| = k = 3n / 2 + 1$ and that~${|\mathcal{V'}| = |\mathcal{V}| - (n/2 + 1) = n}$.
	Now, consider some $V' \in \mathcal{V}'$.
	We have ${V' \subseteq X_i}$ or ${V' \subseteq Y_i}$ for some~$i \in [n]$ by construction.
	Since~$|X_i| = |Y_i| = 2$ for all~$i \in [n]$, we have $|V'| \le 2$ and thus~${\left| \bigcup_{V' \in \mathcal{V}'} V' \right| \le 2 n}$.
	Moreover, since there are~$n/2 + 1 + 4n = 9n/2 + 1$~vertices in~$G$, we have~${\left| \bigcup_{V' \in \mathcal{V}'} V' \right| = |V| - |V_z| - |L| \ge 2n}$.
	Hence,~$|\bigcup_{V' \in \mathcal{V}'} V'| = 2n$ and thus, for each part~$V' \in \mathcal{V'}$, it holds that~${|V'| = 2}$ yielding~$V' = X_i$ or~$V' = Y_i$ for some~${i \in [n]}$.
	Let~${I_x \defeq \{i \in [n] \mid X_i \in \mathcal{V'} \}}$ and~${I_y \defeq \{ i \in [n] \mid Y_i \in \mathcal{V'} \}}$.
	Since all~$X_i$ are $q$-colored and all~$Y_i$ are $r$-colored, we have~${|I_x| \le n / 2}$ and~$|I_y| \le n / 2$.
	Then, since~$|I_x| + |I_y| = |\mathcal{V}| = n$, we obtain~${|I_x| = |I_y| = n / 2}$.
	
	The total weights of vertices of color $q$ and $r$ in $V_z$ are
	\begin{align}
		w_q(V_z) &= \sum_{i \in I_x} w(x_i^q) + \sum_{i \in I_y} w(y_i^q) = Mn + \sum_{i \in I_x} a_i + N \left( \sum_{i \in I_x} 2^i- \sum_{i \in I_y} 2^i \right) \label{eq:wq} \text{ and }
	\end{align}
	\begin{align}
		w_r(V_z) &= \sum_{i \in I_x} w(x_i^r) + \sum_{i \in I_y} w(y_i^r) = Mn + s - \sum_{i \in I_y} a_i + N \left( \sum_{i \in I_y} 2^i - \sum_{i \in I_x} 2^i \right), \label{eq:wr}
	\end{align}
	respectively.
	Now, assume for the sake of contradiction that~$I_x \ne I_y$.
	Then, there exists an index~${i_{\max} \defeq \max \{ (I_x \setminus I_y) \cup (I_y \setminus I_x) \}}$.
	If~$i_{\max} \in I_x$, then each element in $I_y \setminus I_x$ is smaller than~$i_{\max}$, and hence
	\begin{align}
		\sum_{i \in I_x} 2^i - \sum_{i \in I_y} 2^i 
		= \sum_{i \in I_x \setminus I_y} 2^i - \sum_{i \in I_y \setminus I_x} 2^i
		\ge 2^{i_{\max}} - \sum_{i \in [i_{\max} - 1]} 2^i = 2. \label{eq:sumi}
	\end{align}
	Combining \cref{eq:wq,eq:sumi} yields $${w_q(V_z) \ge Mn + \sum_{i \in I_x} a_i + 2N \ge Mn + 2N > w_p(V_z)},$$ which is a contradiction to~$V_z$ being uniquely~$p$-colored.
	We analogously obtain a contradiction for $i_{\max} \in I_y$ and thus it holds that~$I_x = I_y$.
	Observe that for~$I \defeq I_x = I_y$ \cref{eq:wq} implies~$w_q(V_z) = Mn + \sum_{i \in I} a_i$ and \cref{eq:wr} implies~$w_r(V_z) = Mn + s - \sum_{i \in I} a_i$.
	
	Since~$w_q(V_z) < w_p(V_z)$ and~$w_r(V_z) < w_p(V_z)$, we obtain $$w_q(V_z) = w_r(V_z) = Mn + s / 2$$ and thus~${\sum_{i \in I} a_i = s / 2}$.
	Consequently, $I$~is a solution to the original instance of \textsc{Partition}.
\end{proof}

We continue with a complexity analysis for the case~${|C|=2}$. 
Note that \GerryJ{} on trees is pseudo-polynomial-time solvable for any constant $|C|$ (and thereby for~$|C| = 2$) \cite{IKK19}.
To complement this result and also \cref{TreeNPhard}, we next show that for~$|C|=2$ there is a polynomial-time algorithm for trees, adapting a pseudo-polynomial-time algorithm of Ito et al.\;\cite[Theorem~4.5]{IKK19}.
We thus obtained a dichotomy with respect to~$|C|$. 
The key difference is that we only store the maximum winning margin of the target color over the other color.

\begin{restatable}{proposition}{ceqtwo}
	\label{thm:treeP}
	For ${|C| = 2}$, {\GerryJ{}} restricted to trees can be solved in $O(n^3)$~time. 
\end{restatable}

\begin{proof}
	We assume that $C = \{ p, q \}$, where $p$ is the target color.
	We provide a polynomial-time algorithm for rooted trees.
	Note that any unrooted tree can be regarded as rooted by choosing an arbitrary vertex as its root.
	Let $r$ be the root of the input graph $G = (V, E)$.
	For each vertex~$v \in V$, let~$G_v$ be the subtree of~$G$ rooted at~$v$.

	Our algorithm is based on dynamic programming.
	We iteratively find partial solutions (which will be defined shortly), starting from the leaves until reaching the root.
	Let $u$ be some vertex of~$G$ and let $v_1, \dots, v_s$ be the children of $u$.
	Let $G_u^0$ be a rooted tree on a single vertex $u$, and for each~$i \in [s]$ let $G_u^i$ be the rooted subtree of $G_u$ induced by $u$ and the vertices of $G_{v_1}, \dots, G_{v_i}$.
	For each vertex~$u$, each~$i \in [s]$, and each~$k' \in [\min(n_u^i,k)]$ (where $n_u^i$ denotes the number of vertices in~$G_u^i$), we define $L^i_{u} \colon \mathbb{N} \to \mathbb{N}$ such that~$L^i_u(k')$ is the maximum number of~$p$-colored parts among all partitions~$\mathcal{V} = (V_1, \dots, V_{k'})$ of the vertices in $G^i_u$.
	Therein, we require that~$G[V_j]$ is connected for each~$j \in [k']$ and that~$u \in V_{k'}$.
	Moreover, we say that the color of~$V_{k'}$ is still undecided as~$V_{k'}$ is the only part that is still connected to the rest of the graph (through the parent of~$u$) and therefore we neglect~$V_{k'}$ when computing~$L^i_u(k')$.
	Further, for each vertex~$u$ and each~$k' \in [\min(n_u^i,k)]$ let
	\begin{align*}
		W^i_u(k') \defeq \max (w_p(V_{k'}) - w_q(V_{k'}))
	\end{align*}
	be the maximum winning margin of~$p$ over~$q$ in~$V_{k'}$ over all~$k'$-partitions~$\mathcal{V}$ of the vertices in $G_u^i$ maximizing~$L^i_u(k')$.
	Observe that a given instance is a yes-instance if and only if $L_r^s(k) + \mathbf{1}[W_r^s(k) > 0] > k / 2$, where~$s$ is the number of children of the root~$r$ and where~$\mathbf{1}[x]$ equals one if the predicate~$x$ is true and zero otherwise.

	We next show how to compute the values of~$L^i_u(k')$ and~$W_u^i(k')$.
	We first initialize the values of $L_u^i(k')$ and~$W_u^i(k')$ for $k' = 1$ as follows:
	\begin{align*}
		L_u^i(1)  &\defeq 0, \\
		W_u^i(1)  &\defeq w_p(V_u^i) - w_q(V_u^i),
	\end{align*}
	where~$V_u^i$ is the set of vertices of~$G_u^i$.
	Note that~$i = 0$ implies~$k' = 1$ as~$G_u^i$ only contains a single vertex.
	Thus, it only remains to compute the values of~$L_u^i$ and~$W_u^i$ for $i > 0$ and~$k' > 1$.
	For a partition $\mathcal{V} = \{ V_1, \dots, V_{k'} \}$ of the vertices of~$G_i^u$ that maximizes $L_u(k')$ and $W_u(k')$, we have two cases: $v_i \in V_{k'}$ or $v_i \notin V_{k'}$.
	If $v_i \notin V_{k'}$, then the edge~$\{u,v_i\}$ is removed and the maximum number of uniquely $p$-colored subsets is the maximum sum of $p$-colored subsets in~$G_u^{i-1}$ and~$G_{v_i}$, that is,
	\begin{align}
		\label{eq:dp2}
		N_u^{i}(k') \defeq \max_{j \in [k' - 1]} & L_u^{i-1}(j) + L_{v_{i}}(k' - j) + \mathbf{1}[W_{v_{i}}(k' - j) > 0].
	\end{align}
	Observe that since~$v_i \notin V_{k'}$, we now count the part that contains~$v_i$ and therefore include the last summand.
	Otherwise (that is, $v_i \in V_{k'}$), then the maximum number of uniquely~\mbox{$p$-colored} subsets is
	\begin{align}
		\label{eq:dp1}
		M_u^{i}(k') \defeq \max_{j \in [k']} L_u^{i-1}(j) + L_{v_i}(k' - j + 1).
	\end{align}
	For the computation of~$L^i_u(k')$ and~$W^i_u(k')$, we have
	\begin{align*}
		L_u^{i}(k') &\defeq \max(N_u^{i}(k'), M_u^{i}(k')), \text{ and}\\
		W_u^{i}(k') &\defeq
		\begin{cases}
			W_u^{i-1}(j) & \text{ if } N_u^{i}(k') > M_u^{i}(k'),\\
			W_u^{i-1}(j') + W_{v_i}(k' - j' + 1) & \text{ if } N_u^{i}(k') < M_u^{i}(k'),\\
			\max(W_u^{i-1}(j), W_u^{i-1}(j') + W_{v_i}(k' - j' + 1) & \text{ if } N_u^{i}(k') = M_u^{i}(k').
		\end{cases}
	\end{align*}
	Here, $j$ and $j'$ are the indices maximizing the terms in Definitions (\ref{eq:dp1}) and~(\ref{eq:dp2}), respectively.
	Regarding the running time, observe that we compute~$O(k \cdot \deg(v))$ table entries for each vertex~$v$.
	Since a tree has $n-1$~edges, we compute by the handshaking lemma in total~$O(nk)$ table entries and computing each table entry requires to sum up at most~$n$ values (weights of vertices or precomputed table entries).
	Thus, the total running time is~${O(n^2 \cdot k) \subseteq O(n^3)}$.
\end{proof}

Finally, we bridge the gap for trees of fixed diameter by generalizing the known polynomial-time algorithm for trees of diameter two~\cite{IKK19} to trees of diameter three.
It is also known that \GerryJ{} on trees of diameter four remains NP-hard \cite{IKK19}.
 
The key observation is that a tree of diameter three can be obtained from two stars by adding an edge between their centers.
Our algorithm then adapts a polynomial-time algorithm for stars~\cite{IKK19}.

\begin{restatable}{proposition}{diamthree}
	\label{thm:diam3}
	For trees of diameter three, \GerryJ{} is solvable in ${O(|C|^2 \cdot n^5)}$~time.
\end{restatable}

\begin{proof}
	First, observe that a tree of diameter three is the same as two stars whose centers are connected by an edge~$e$.
	Let~$r_1$ and~$r_2$ be the two centers of the stars~($e = \{r_1,r_2\}$).
	Our algorithm distinguishes between two cases: (i)~$r_1$ and~$r_2$ belong to the same part in a solution~$\mathcal{V}$, and (ii) they belong to two different parts.\footnote{Technically, our algorithm computes whether there exists a solution for each of the two cases and reports a solution if it finds a solution for any of the two cases or rejects the input if it does not find a solution for any of the two cases.}
	The subalgorithm for case~(i) is completely analogous to the algorithm for \GerryJ{} on stars (trees of diameter two) by Ito et al.\;\cite{IKK19}.
	We will present the whole algorithm for the sake of completeness.
	It will also be helpful in understanding the subalgorithm for case~(ii) (which is an adaptation of the first subalgorithm).
	Both subalgorithms are based on the observation that each part of~$\mathcal{V}$ not containing~$r_1$ or~$r_2$ only consists of a single vertex.

	We start with presenting the subalgorithm for case~(i).
	The algorithm guesses\footnote{Whenever we ``guess'' something, we iterate over all possible cases and test whether this iteration yields a solution. If any iteration yields a solution, then we refer to this iteration in the proof.} a color~$q^*$ such that the part~$V_\ell \in \mathcal{V}$ with~$r_1,r_2 \in V_\ell$ is~$q^*$-colored (uniquely~$p$-colored if~$q^* = p$).
	Moreover, the algorithm guesses the numbers~$\alpha(p)$ and~$\alpha(q^*)$ of~$p$-colored and~$q^*$-colored leaves that are not contained in~$V_\ell$ (those leaves form their own parts in~$\mathcal{V}$).
	Let $x$ be the number of uniquely $p$-colored parts in~$\mathcal{V}$ (note that $x = \alpha(p) + 1$ if $q^* = p$ and $x = \alpha(p)$ otherwise).
	Note that for each color~$q \neq p$ we have to guarantee that there are at most~$x - 1$ parts that are~$q$-colored in~$\mathcal{V}$.
	Let~$n_{q^*}$ be the number of~$q^*$-colored leaves.
	As proven by Ito et al.\;\cite[Lemma 4.3]{IKK19}, one can assume that any $q^*$-colored leaf in $V_\ell$ is at least as heavy as the ones  not in $V_\ell$.
	So we can assume that $w_{q^*}(V_\ell)$ is the sum of the heaviest~$n_{q^*} - \alpha(q^*)$ $q^*$-colored leaves plus $w_{q^*}(\{ r_1, r_2 \})$.
	Similarly (also shown by Ito et al.\;\cite[Lemma 4.3]{IKK19}), we can also assume that any $q$-colored leaf (for $q \ne q^*$) not in $V_\ell$ is at least as heavy as the ones in~$V_\ell$.
	For each color~$q \in C \setminus \{p,q^*\}$, let~$\beta(q)$ be the smallest number of $q$-colored leaves that cannot be included in~$V_\ell$.
	By definition, $\beta(q)$ is the minimum number such that the sum of weights of all but the~$\beta(q)$ heaviest $q$-colored leaves is at most~$w_{q^*}(V_\ell) - w_{q}(\{ r_1, r_2 \})$ (strictly less if $q^* = p$).
	Finally, we verify the following:
	\begin{itemize}
		\item For each $q \in C \setminus \{ p, q^* \}$ it holds that~$\beta(q) < x$ (and~$\beta(q) + 1 < x$ if removing the~$\beta(q)$ heaviest leaves results in~$V_\ell$ being~$q$-colored).
		\item The values of $\beta(q)$ for all colors $q \in C \setminus \{ p, q^* \}$ plus~$\alpha(q^*)$ and~$\alpha(p)$ sums up to at most~$k$.
		\item There are enough leaves that are not~$p$-colored or~$q^*$-colored that can be removed after removing~$\beta(q)$ vertices of each color~$q$ to achieve~$k$ parts.
	\end{itemize}
	If so, then we remove the remaining vertices arbitrarily to obtain a solution.
	
	We continue with the subalgorithm for the case (ii).
	Although the algorithm is somewhat similar to the previous case, we compute the values (namely, $\alpha$ and~$\beta$) for each of the two stars separately.
	We assume without loss of generality that~$r_1 \in V_1 \in \mathcal{V}$ and~$r_2 \in V_2 \in \mathcal{V}$.
	We first guess two colors~$q^*_1$ and $q^*_2$ such that the part~$V_i \in \mathcal{V}$ with~$i \in \{1,2\}$ is~$q^*_i$-colored (uniquely~$p$-colored if~${q^*_i = p}$).
	We then guess the numbers~$\alpha_1(p)$,~$\alpha_2(p)$,~$\alpha_1(q^*_1)$ and~$\alpha_2(q^*_2)$ of~$p$-colored and~$q^*_i$-colored leaves adjacent to~$r_i$ that are not contained in~$V_i$.
	From these guesses, we analogously compute the number~$x$ of uniquely~$p$-colored parts and for~$i \in \{1,2\}$ the sum~$w_{q_i^*}(V_i)$ of weights of all~$q^*_i$-colored vertices in~$V_i$.
	Next, we compute for each $i \in \{1,2\}$ and each color~$q \notin \{p,q^*_i\}$ the minimum number~$\beta_i(q)$ of leaves that are adjacent to~$r_i$ but that cannot be contained in~$V_i$.
	Now it remains to check for the special case whether after removing~$\beta_i(q)$ vertices the part~$V_i$ is~$q$-colored and to check whether~$\beta_1(q) + \beta_2(q) < x$.
	Finally, we verify whether the sum over all~$\beta$-values plus~$\alpha_i(q^*_i)$ and~$\alpha_i(p)$ is at most~$k$ and whether there are enough leaves that are not~$p$ or~$q^*_i$-colored that can be removed after removing~$\beta_i(q)$ vertices of each color~$q$ from the star with center~$r_i$.
	If so, then we remove the remaining vertices arbitrarily to obtain a solution.
	
	Finally, it remains to analyze the running time.
	First, we sort all leaves of each color by their respective weights in~$O(n \log n)$ time.
	Second we guess the color(s)~$q_{(i)}^*$ and the numbers~$\alpha_{(i)}(p)$ and~$\alpha_{(i)}(q^*_{(i)})$.
	There are~$O(|C|^2 \cdot n^4)$ possible guesses.
	Third, we compute~$w_{q_{(i^*)}}(V_{(i)})$ and~$\beta_{(i)}(q)$ for all~$q \notin \{p,q^*_{(i)}\}$ in overall linear time.
	The final checks also take~$O(n)$ time, resulting in the claimed running time~$O(|C|^2 \cdot n^5)$.
\end{proof}

\Cref{thm:diam3} yields a complexity dichotomy for trees with respect to the diameter parameter.
Clearly, our polynomial-time solvability is mainly of classification nature; it remains a future task to lower the degree in the polynomial of the running time.

\section{Conclusion}
Answering open questions of Ito et al.\;\cite{IKK19} in the negative,
we presented an NP-hardness result on paths and a weak NP-hardness result on trees.
Now, one may claim that the computational complexity of \GerryJ{} restricted to 
paths and trees is well-understood.
The results indicate that, through the lens of worst-case complexity analysis, 
\GerryJ{} is extremely hard. Indeed, from our findings one can also deduce negative 
results in terms of parameterized complexity analysis, that is NP-hardness for constant values of each single graph parameter vertex cover number, maximum leaf number, 
and vertex deletion number to cliques. In parameterized complexity theory, these are 
among the ``weakest'' parameters.

As previous work, we focused on the Plurality voting rule, leaving open
to study \GerryJ{} also for other voting rules.
Moreover, we focused on theoretical results. Since worst-case intractability 
is clearly no shield against susceptibility of real-world instances to gerrymandering,
following the example of Cohen-Zemach et al.\;\cite{ZLR18}
it may be promising to investigate the empirical issues. 

\bibliographystyle{plainnat}
\bibliography{literature.bib}

\end{document}